\documentclass[11pt]{scrartcl}

\usepackage{url}
\usepackage[hidelinks]{hyperref}
\usepackage[utf8]{inputenc}
\usepackage[small]{caption}
\usepackage{graphicx}
\usepackage{amsmath}
\usepackage{booktabs}
\usepackage{algorithm}
\usepackage{algorithmic}
\usepackage{amsthm}
\usepackage{amssymb}
\usepackage{tikz}

\usepackage{color}

\newtheorem{theorem}{Theorem}

\newtheorem{lemma}[theorem]{Lemma}

\theoremstyle{definition}

\usepackage{natbib}

\title{On the Number of \\Almost Envy-Free Allocations}

\author{
Warut Suksompong\\University of Oxford
}

\date{\vspace{-3ex}}

\begin{document}

\maketitle

\begin{abstract}
Envy-freeness is a standard benchmark of fairness in resource allocation. 
Since it cannot always be satisfied when the resource consists of indivisible items even when there are two agents, the relaxations \emph{envy-freeness up to one item (EF1)} and \emph{envy-freeness up to any item (EFX)} are often considered.
We establish tight lower bounds on the number of allocations satisfying each of these benchmarks in the case of two agents.
In particular, while there can be as few as two EFX allocations for any number of items, the number of EF1 allocations is always exponential in the number of items.
Our results apply a version of the vertex isoperimetric inequality on the hypercube and help explain the large gap in terms of robustness between the two notions.
\end{abstract}

\section{Introduction}

The allocation of scarce resources to interested agents is a task that arises commonly in our everyday lives.
Indeed, course slots need to be allocated to university students, grant funding to researchers, and personnel to organizations, to name but a few examples.
A chief concern when allocating resources is \emph{fairness}---we want all agents to feel that they receive a fair share of the resources.
Several benchmarks of fairness have been proposed in the fair division literature.
One of the most fundamental benchmarks is \emph{envy-freeness}, which means that every agent should receive their first choice among all of the allocated bundles.
Put differently, no agent should have reason to envy any other agent.

While envy-freeness is a compelling fairness benchmark, it suffers from the setback that it cannot always be satisfied when the resource consists of indivisible items such as books, course slots, or personnel.
This is evident in the simple instance where there are two agents and only one indivisible item to be allocated.
Consequently, two natural relaxations of envy-freeness have been considered.
\emph{Envy-freeness up to one item (EF1)} demands that any envy that one agent has towards another agent can be eliminated by removing a single item of our choice from the latter agent's bundle, while the stronger notion of \emph{envy-freeness up to any item (EFX)} requires that the envy can be eliminated by removing \emph{any} single item from the latter agent's bundle.
An allocation that satisfies EF1 and EFX always exists when the items are divided between two agents with arbitrary valuations over subsets of items \citep{LiptonMaMo04,PlautRo18}.\footnote{\label{footnote:monotonic}The valuations are assumed to be \emph{monotonic}, i.e., an agent's value for a set cannot decrease when an item is added to the set.
This assumption is usually made without loss of generality, since agents can freely dispose of items that they do not like.}

The notions of EF1 and EFX serve as fairness criteria that can always be fulfilled when allocating indivisible items between two agents.
On the one hand, given that EFX provides a stronger fairness guarantee than EF1, one might view it to be the appropriate criterion in this setting.
On the other hand, a series of recent work has shown that EF1 is a more robust criterion than EFX in several ways:
\begin{itemize}
\item An EF1 allocation can be computed using a number of queries that is only logarithmic in the number of items \citep{OhPrSu19}.
On the contrary, computing an EFX allocation takes an exponential number of queries in the worst case \citep{PlautRo18}, and a linear number of queries even when the valuations are additive \citep{OhPrSu19}.
\item There always exists a \emph{balanced} EF1 allocation (that is, the numbers of items that the two agents receive differ by no more than $1$), while an EFX allocation may necessarily be highly unbalanced \citep{KyropoulouSuVo19}.
Moreover, when the agents' valuations are additive, an EF1 allocation fulfilling a set of cardinality constraints can also be found \citep{BiswasBa18}.
\item With additive valuations, an EF1 allocation satisfying the economic efficiency condition of \emph{Pareto optimality} always exists, i.e., no other allocation makes one agent better off without making the other agent worse off \citep{BarmanKrVa18,CaragiannisKuMo16}. On the other hand, a Pareto-optimal EFX allocation may not exist \citep{PlautRo18}.
In addition, requiring EF1 leads to a smaller social welfare loss than EFX in the worst case \citep{BeiLuMa19}.
\item If the items lie on a line, an EF1 allocation in which each agent receives a contiguous block of items is guaranteed to exist \citep{BiloCaFl19,Suksompong19}; this is not the case for EFX.
More generally, when the items lie on a graph, the existence of an EF1 allocation in which every agent receives a connected subgraph can be guaranteed for a large class of graphs \citep{BiloCaFl19}, whereas an EFX allocation always exists only if the graph is complete \citep{BeiIgLu19}.
\item Beyond the setting of two agents, an EF1 allocation can always be found for any number of agents \citep{LiptonMaMo04}, whereas the existence question for EFX remains intriguingly open when there are at least four agents with additive valuations, or three agents with non-additive valuations \citep{PlautRo18,CaragiannisKuMo16,AmanatidisBiFi20,AmanatidisNtMa20,ChaudhuryGaMe20}.
\end{itemize}
In this note, we provide an explanation to the robustness of EF1 by deriving tight lower bounds on the number of EF1 and EFX allocations for arbitrary valuations of the two agents.
Specifically, while there can be as few as two EFX allocations regardless of the number of items, the number of EF1 allocations is---quite surprisingly---always exponential in the number of items.\footnote{We also remark that several extensions of EF1 have been recently studied, including \emph{weighted EF1} \citep{ChakrabortyIgSu20}, \emph{typewise EF1} \citep{BenabbouChEl19}, and \emph{democratic EF1} \citep{SegalhaleviSu19}.}

\section{The Bounds}

In our formal model, there are two agents, who we sometimes call Alice and Bob, and a set $M$ of $m\geq 1$ items. 
A subset of $M$ is referred to as a \emph{bundle}.
For $i\in\{1,2\}$, agent $i$ has a nonnegative value $u_i(M')$ for each bundle $M'\subseteq M$.
We assume without loss of generality that $u_i(\emptyset) = 0$ and, as is commonly done, that the valuations are \emph{monotonic}, i.e., $u_i(A)\leq u_i(B)$ for any $A\subseteq B\subseteq M$.
A bundle $M'$ is said to be
\begin{itemize}
\item \emph{envy-free up to one item (EF1)} for agent $i$, if  either $u_i(M')\geq u_i(M\backslash M')$, or there exists an item $j\in M\backslash M'$ such that $u_i(M') \geq u_i(M\backslash(M'\cup\{j\}))$;
\item \emph{envy-free up to any item (EFX)} for agent $i$, if for each item $j\in M\backslash M'$, we have $u_i(M') \geq u_i(M\backslash(M'\cup\{j\}))$.
\end{itemize}

An \emph{allocation} is an ordered partition of $M$ into two sets $(M_1,M_2)$, where bundle $M_i$ is given to agent $i$.
The allocation is said to be EF1 (resp. EFX) if for each $i\in\{1,2\}$, bundle $M_i$ is EF1 (resp. EFX) for agent $i$.
Moreover, we say that an unordered partition of $M$ into two sets $(M',M'')$ is EF1 (resp. EFX) for agent $i$ if both $M'$ and $M''$ are EF1 (resp. EFX) for the agent.

We now derive the lower bounds on the number of EF1 and EFX allocations.
We start with the stronger notion, EFX.
For any number of agents with identical valuations, \cite{PlautRo18} showed that an EFX allocation always exists.
This implies that in our setting with two agents, each agent has an EFX partition of the items, a fact that will be useful for our proof.

\begin{theorem}
For two agents with arbitrary monotonic valuations, the number of EFX allocations is at least $2$.
Moreover, this bound is tight for any number of items $m$.
\end{theorem}

\begin{proof}
We let each agent propose an (unordered) EFX partition.
If the two proposed partitions coincide, we have two EFX allocations by assigning either bundle to either agent.
If they differ, letting Bob choose a preferred bundle from Alice's partition yields an EFX allocation, and letting Alice choose a preferred bundle from Bob's partition yields another EFX allocation different from the first one.
Hence we can find two EFX allocations in both cases.

To show that the bound is tight, assume that the agents have identical valuations with value $1$ for each of the first $m-1$ items, and value $m$ for the last item. 
The valuations are additive: the value for a bundle of items is simply the sum of the values of the individual items in the bundle.
It is clear that all of the first $m-1$ items must be given to the agent who does not receive the last item in order for the allocation to be EFX.
Hence there are only two EFX allocations in this instance.
\end{proof}

Next, we move on to EF1, for which deriving a tight bound is considerably more challenging.

\begin{theorem}
\label{thm:EF1}
For two agents with arbitrary monotonic valuations, the number of EF1 allocations is at least $f_{\text{EF1}}(m)$, where
\[
f_{\text{EF1}}(m) := 
\begin{cases}
\dbinom{m}{m/2} & \text{ if } m \text{ is even;} \\
2\cdot\dbinom{m-1}{(m-1)/2} & \text{ if } m \text{ is odd.}
\end{cases}
\]
Moreover, this bound is tight for any number of items $m$.
\end{theorem}

To prove this theorem, we will rely on two combinatorial results.
The first result is a version of the vertex isoperimetric problem on the hypercube, where we want to choose a certain number of vertices on a hypercube so that the ``vertex boundary'' of this set is as small as possible.
For a finite set $X$, denote its power set by $2^X$.
A \emph{set system} is a subset of $2^X$.
The \emph{Hamming distance} between two sets $A,B\subseteq X$ is the size of their symmetric difference: 
$$d(A,B) = |(A\backslash B)\cup(B\backslash A)|,$$
and the Hamming distance between two nonempty set systems $\mathcal{A},\mathcal{B}\subseteq 2^X$ is
$$d(\mathcal{A},\mathcal{B}) = \min_{A\in\mathcal{A},B\in\mathcal{B}}d(A,B).$$
\emph{The Hamming ball} of center $A\subseteq X$ and radius $r\in\mathbb{N}\cup\{0\}$ is
$$H_r(A) := \{B\subseteq X\mid d(A,B)\leq r\}.$$
A set system $\mathcal{B}\subseteq 2^X$ is called \emph{a Hamming ball} of center $A\subseteq X$ and radius $r\in \mathbb{N}$ if
$$H_{r-1}(A)\subseteq\mathcal{B}\subseteq H_r(A).$$
Note the difference between ``the Hamming ball'' and ``a Hamming ball''.

\begin{lemma}[\cite{Bollobas86,Calabro04}]
\label{lem:hamming}
Let $X$ be a finite set, and let $\mathcal{A},\mathcal{B}\subseteq 2^X$ be nonempty set systems. 
There exist a Hamming ball $\mathcal{A}_0$ with center $X$ and a Hamming ball $\mathcal{B}_0$ with center $\emptyset$ such that $|\mathcal{A}_0|=|\mathcal{A}|$, $|\mathcal{B}_0|=|\mathcal{B}|$, and $d(\mathcal{A}_0,\mathcal{B}_0)\geq d(\mathcal{A},\mathcal{B})$.
\end{lemma}

The second result that we will use in our proof concerns \emph{Sperner families}, i.e., set systems in which no set contains another.
A famous theorem of \cite{Sperner28} states that if $|X|=n$, a Sperner family in $2^X$ can have size at most $\binom{n}{\lfloor n/2\rfloor}$.
\citet[Thm.~2.2]{Bjorner86} provided more precise information regarding the number of sets of different sizes that can appear in a Sperner family.
For any positive integers $n$ and $k$, there is a unique way of writing
$$
n = \binom{a_k}{k} + \binom{a_{k-1}}{k-1}+\dots+\binom{a_i}{i}
$$
so that $a_k>a_{k-1}>\dots>a_i\geq i\geq 1$ are integers. 
We define
$$
\partial_{k-1}(n) = \binom{a_k}{k-1} + \binom{a_{k-1}}{k-2}+\dots+\binom{a_i}{i-1}
$$
and let $\partial_{k-1}(0) = 0$.

\begin{lemma}[\cite{Bjorner86}]
\label{lem:sperner}
Let $X$ be a set with $|X|=n$, and let $c_0,c_1,\dots,c_{n-1}$ be nonnegative integers such that at least one is strictly positive.
There exists a Sperner family in $2^X$ with exactly $c_i$ members of size $i+1$ for $i=0,1,\dots,n-1$ if and only if
$$
\partial_{j+1}(\cdots\partial_{n-2}(\partial_{n-1}(c_{n-1})+c_{n-2})+c_{n-3}\cdots)+c_j \leq \binom{n}{j+1},
$$
where $j$ is the smallest index such that $c_j \neq 0$.
\end{lemma}

With Lemmas~\ref{lem:hamming} and \ref{lem:sperner} in hand, we are now ready to derive the lower bound on the number of EF1 allocations.

\begin{proof}[Proof of Theorem~\ref{thm:EF1}]
We first show tightness of the bound. 
Assume that the agents have identical valuations, and the value for a bundle of items is simply the sum of the values of the individual items in the bundle.
If $m$ is even, suppose that the agents have value $1$ for every item. 
The EF1 allocations are exactly the allocations that assign $m/2$ items to each agent, so there are exactly $\binom{m}{m/2}$ EF1 allocations.
For $m$ odd, suppose that the agents have value $1$ for each of the first $m-1$ items, and value $0$ for the last item.
The first $m-1$ items must be split equally between the two agents, while the last item can go to either agent.
Hence there are exactly $2\cdot\binom{m-1}{(m-1)/2}$ EF1 allocations.

Next, we proceed to establish the bound.
We claim that for each agent, at least $\frac{1}{2}\cdot f_{\text{EF1}}(m)$ partitions are EF1.
Observe that this claim immediately yields the desired result.
To see this, list $\frac{1}{2}\cdot f_{\text{EF1}}(m)$ partitions that are EF1 for Alice, and $\frac{1}{2}\cdot f_{\text{EF1}}(m)$ partitions that are EF1 for Bob.
For a partition that appears in both lists, we can allocate either part to either agent, giving rise to two EF1 allocations.
For a partition that appears in only one list, we let the other agent choose a preferred bundle from the partition, giving rise to one EF1 allocation.
Hence the number of EF1 allocations is at least $\frac{1}{2}\cdot f_{\text{EF1}}(m) + \frac{1}{2}\cdot f_{\text{EF1}}(m) = f_{\text{EF1}}(m)$.

We now restrict our attention to Alice and prove that at least $\frac{1}{2}\cdot f_{\text{EF1}}(m)$ partitions are EF1 for her.
Recall that $M$ denotes the set of all items, and call a bundle $M'\subseteq M$ ``good'' if the partition $(M',M\backslash M')$ is EF1 for Alice, and ``bad'' otherwise.
Our goal is equivalent to showing that at least $f_{\text{EF1}}(m)$ bundles are good.
A bundle can be bad for one of two reasons: either it is ``too large'' (i.e., it is EF1, but its complement is not), or it is ``too small'' (i.e., it is not EF1, but its complement is).
Note that a bundle is too small if and only if its complement is too large.
In addition, it is easy to check from the definition of EF1 that if we remove an item from a bundle that is too small, then the resulting bundle is, as one should expect, also too small.
Likewise, if we add an item to a bundle that is too large, the resulting bundle remains too large.

Consider bundles of items as elements of $2^M$.
We claim that if bundle $A$ is too small and bundle $B$ is too large, then $d(A,B)\geq 2$.
To prove this claim, it suffices to show that if a bundle is too small, then adding one item to it cannot make it too large.
Assume that bundle $A$ is too small, and let $j\in M\backslash A$.
Denoting Alice's valuation by $u$, since $A$ is too small, we have $u(A)<u(M\backslash(A\cup\{j\}))$.
This is equivalent to $u(M\backslash(A\cup\{j\})) > u((A\cup\{j\})\backslash\{j\})$, which means that $u(M\backslash(A\cup\{j\}))$ is not too small.
Hence $A\cup\{j\}$ is not too large, which establishes our claim.

Let $\mathcal{A}$ (resp. $\mathcal{B}$) be a set system consisting of all bundles that are too small  (resp. too large).
It follows from the preceding paragraph that $d(\mathcal{A},\mathcal{B})\geq 2$.
By Lemma~\ref{lem:hamming}, there exist a Hamming ball $\mathcal{A}_0$ with center $M$ and a Hamming ball $\mathcal{B}_0$ with center $\emptyset$ such that $|\mathcal{A}_0|=|\mathcal{A}|$, $|\mathcal{B}_0|=|\mathcal{B}|$, and $d(\mathcal{A}_0,\mathcal{B}_0)\geq d(\mathcal{A},\mathcal{B})\geq 2$.
Moreover, symmetry implies that $|\mathcal{A}|=|\mathcal{B}|$.
The number of good bundles is $
|2^M\backslash(\mathcal{A}\cup\mathcal{B})|
=
|2^M\backslash(\mathcal{A}_0\cup\mathcal{B}_0)|$.

Assume first that $m$ is even. 
Let $r$ be the unique integer such that $H_{r-1}(M)\subsetneq\mathcal{A}_0\subseteq H_r(M)$.
Since $|\mathcal{A}_0|=|\mathcal{B}_0|$, we have $H_{r-1}(\emptyset)\subsetneq\mathcal{B}_0\subseteq H_r(\emptyset)$.
Note that $H_{r-1}(M)$ contains all bundles of size at least $m-(r-1)$, and $H_{r-1}(\emptyset)$ contains all bundles of size at most $r-1$.
If $r\geq m/2$, then since $H_{r-1}(M)\subsetneq\mathcal{A}_0\subseteq H_r(M)$, we have that $\mathcal{A}_0$ contains at least one bundle of size $m-m/2 = m/2$.
Moreover, since $\mathcal{B}_0$ contains $H_{r-1}(\emptyset)$, it contains \emph{all} bundles of size $m/2-1$.
At least one of these bundles has Hamming distance $1$ from a bundle of size $m/2$ in $\mathcal{A}_0$, a contradiction.
So $r<m/2$, which means that no bundle of size $m/2$ belongs to $\mathcal{A}_0\cup\mathcal{B}_0$.
It follows that the number of good bundles is
$$
|2^M\backslash(\mathcal{A}_0\cup\mathcal{B}_0)|
\geq
\binom{m}{m/2} = f_{\text{EF1}}(m).
$$

For the rest of this proof, assume that $m=2s+1$ is odd.
As in the previous paragraph, let $r$ be such that $H_{r-1}(M)\subsetneq\mathcal{A}_0\subseteq H_r(M)$ and $H_{r-1}(\emptyset)\subsetneq\mathcal{B}_0\subseteq H_r(\emptyset)$.
If $r\geq s+1$, then $\mathcal{A}_0$ contains all bundles of size $s+1$ and $\mathcal{B}_0$ contains all bundles of size $s$, contradicting $d(\mathcal{A}_0,\mathcal{B}_0)\geq 2$.
If $r\leq s-1$, then no bundle of size $s$ or $s+1$ belongs to $\mathcal{A}_0\cup\mathcal{B}_0$.
In this case, the number of good bundles is 
$$
|2^M\backslash(\mathcal{A}_0\cup\mathcal{B}_0)|
\geq
2\binom{m}{s}
\geq
2\binom{m-1}{s}
= f_{\text{EF1}}(m).
$$
Suppose now that $r=s$.
It suffices to prove that $\mathcal{A}_0$ contains at most $\binom{m-1}{s-1}$ bundles of size $s+1$, and $\mathcal{B}_0$ contains at most $\binom{m-1}{s-1}$ bundles of size $s$.
Indeed, this would imply that the number of good bundles is 
$$
|2^M\backslash(\mathcal{A}_0\cup\mathcal{B}_0)|
\geq
2\left(\binom{m}{s} - \binom{m-1}{s-1}\right)
=
2\binom{m-1}{s}
= f_{\text{EF1}}(m).
$$

Assume for contradiction that $\mathcal{A}_0$ contains $t > \binom{m-1}{s-1}$ bundles of size $s+1$.
By symmetry, $\mathcal{B}_0$ contains $t$ bundles of size $s$.
Since $d(\mathcal{A}_0,\mathcal{B}_0)\geq 2$, none of these bundles in $\mathcal{B}_0$ can be a subset of any of the $t$ bundles in $\mathcal{A}_0$.
This means that the $2t$ bundles together form a Sperner family in $2^M$.
Applying Lemma~\ref{lem:sperner}, we find that
$$\binom{m}{s}\geq \partial_s(t) + t > \partial_s(t) + \binom{m-1}{s-1},$$
or $\partial_s(t) < \binom{m}{s} - \binom{m-1}{s-1} =  \binom{m-1}{s}$.
Observe that $\partial_s\left(\binom{m-1}{s-1}\right) = \partial_s\left(\binom{m-1}{s+1}\right) = \binom{m-1}{s}$.
Since $t > \binom{m-1}{s-1}$, in order to obtain the desired contradiction, we only need to show that $\partial_{k-1}(n)$ is a non-decreasing function of $n$ for any fixed $k$.

To prove this statement, we proceed by induction on $k$.
The base case $k=1$ is trivial.
Assume that the statement holds up to $k-1$, and let $n_1> n_2$ be two positive integers.
Write
$$
n_1 = \binom{a_k}{k} + \binom{a_{k-1}}{k-1}+\dots+\binom{a_i}{i}\text{ and }
\partial_{k-1}(n_1) = \binom{a_k}{k-1} + \binom{a_{k-1}}{k-2}+\dots+\binom{a_i}{i-1}
$$
with $a_k>a_{k-1}>\dots>a_i\geq i\geq 1$, and
$$
n_2 = \binom{b_k}{k} + \binom{b_{k-1}}{k-1}+\dots+\binom{b_j}{j}\text{ and }
\partial_{k-1}(n_2) = \binom{b_k}{k-1} + \binom{b_{k-1}}{k-2}+\dots+\binom{b_j}{j-1}
$$
with $b_k>b_{k-1}>\dots>b_j\geq j\geq 1$.
If $a_k = b_k$, the claim follows from the induction hypothesis applied to the two integers $n_1-\binom{a_k}{k} > n_2-\binom{a_k}{k}$.
If $a_k > b_k$, we have
\begin{align*}
\partial_{k-1}(n_1)
&\geq
\binom{a_k}{k-1} \\
&\geq \binom{b_k+1}{k-1} \\
&= \binom{b_k}{k-1} + \binom{b_k}{k-2} \\
&= \binom{b_k}{k-1} + \binom{b_k-1}{k-2} + \binom{b_k-1}{k-3} \\
&= \dots \\
&= \binom{b_k}{k-1} + \binom{b_k-1}{k-2} +\dots + \binom{b_k-k+j+1}{j} + \binom{b_k-k+j+1}{j-1} \\
&\geq \binom{b_k}{k-1} + \binom{b_{k-1}}{k-2} +\dots + \binom{b_{j+1}}{j} + \binom{b_{j+1}}{j-1} \\
&\geq \binom{b_k}{k-1} + \binom{b_{k-1}}{k-2} +\dots + \binom{b_{j+1}}{j} + \binom{b_j}{j-1} \\
&= \partial_{k-1}(n_2).
\end{align*}
Else, $a_k<b_k$. A similar chain of inequalities as above shows that $n_2\geq n_1$, which is impossible.
It follows that $\partial_{k-1}(n_1) \geq \partial_{k-1}(n_2)$, concluding the induction and our proof.
\end{proof}

\section*{Acknowledgments}
The author thanks Dominik Peters for useful discussions, the MathOverflow community for technical help, as well as an anonymous reviewer for helpful comments, and acknowledges support from the European Research Council (ERC) under grant number 639945 (ACCORD).

\bibliographystyle{named}
\bibliography{main}

\begin{thebibliography}{}

\bibitem[\protect\citeauthoryear{Amanatidis \bgroup \em et al.\egroup
  }{2020a}]{AmanatidisBiFi20}
Georgios Amanatidis, Georgios Birmpas, Aris Filos-Ratsikas, Alexandros
  Hollender, and Alexandros~A. Voudouris.
\newblock Maximum {N}ash welfare and other stories about {EFX}.
\newblock In {\em Proceedings of the 29th International Joint Conference on
  Artificial Intelligence (IJCAI)}, 2020.
\newblock Forthcoming.

\bibitem[\protect\citeauthoryear{Amanatidis \bgroup \em et al.\egroup
  }{2020b}]{AmanatidisNtMa20}
Georgios Amanatidis, Apostolos Ntokos, and Evangelos Markakis.
\newblock Multiple birds with one stone: Beating 1/2 for {EFX} and {GMMS} via
  envy cycle elimination.
\newblock In {\em Proceedings of the 34th AAAI Conference on Artificial
  Intelligence (AAAI)}, 2020.
\newblock Forthcoming.

\bibitem[\protect\citeauthoryear{Barman \bgroup \em et al.\egroup
  }{2018}]{BarmanKrVa18}
Siddharth Barman, Sanath~Kumar Krisnamurthy, and Rohit Vaish.
\newblock Finding fair and efficient allocations.
\newblock In {\em Proceedings of the 19th ACM Conference on Economics and
  Computation (EC)}, pages 557--574, 2018.

\bibitem[\protect\citeauthoryear{Bei \bgroup \em et al.\egroup
  }{2019a}]{BeiIgLu19}
Xiaohui Bei, Ayumi Igarashi, Xinhang Lu, and Warut Suksompong.
\newblock The price of connectivity in fair division.
\newblock {\em CoRR}, abs/1908.05433, 2019.

\bibitem[\protect\citeauthoryear{Bei \bgroup \em et al.\egroup
  }{2019b}]{BeiLuMa19}
Xiaohui Bei, Xinhang Lu, Pasin Manurangsi, and Warut Suksompong.
\newblock The price of fairness for indivisible goods.
\newblock In {\em Proceedings of the 28th International Joint Conference on
  Artificial Intelligence (IJCAI)}, pages 81--87, 2019.

\bibitem[\protect\citeauthoryear{Benabbou \bgroup \em et al.\egroup
  }{2019}]{BenabbouChEl19}
Nawal Benabbou, Mithun Chakraborty, Edith Elkind, and Yair Zick.
\newblock Fairness towards groups of agents in the allocation of indivisible
  items.
\newblock In {\em Proceedings of the 28th International Joint Conference on
  Artificial Intelligence (IJCAI)}, pages 95--101, 2019.

\bibitem[\protect\citeauthoryear{Bil\`{o} \bgroup \em et al.\egroup
  }{2019}]{BiloCaFl19}
Vittorio Bil\`{o}, Ioannis Caragiannis, Michele Flammini, Ayumi Igarashi,
  Gianpiero Monaco, Dominik Peters, Cosimo Vinci, and William~S. Zwicker.
\newblock Almost envy-free allocations with connected bundles.
\newblock In {\em Proceedings of the 10th Innovations in Theoretical Computer
  Science Conference (ITCS)}, pages 14:1--14:21, 2019.

\bibitem[\protect\citeauthoryear{Biswas and Barman}{2018}]{BiswasBa18}
Arpita Biswas and Siddharth Barman.
\newblock Fair division under cardinality constraints.
\newblock In {\em Proceedings of the 27th International Joint Conference on
  Artificial Intelligence (IJCAI)}, pages 91--97, 2018.

\bibitem[\protect\citeauthoryear{Bj\"{o}rner}{1986}]{Bjorner86}
Anders Bj\"{o}rner.
\newblock Face numbers of complexes and polytopes.
\newblock In {\em Proceedings of the International Congress of Mathematicians
  (ICM)}, pages 1408--1418, 1986.

\bibitem[\protect\citeauthoryear{Bollob\'{a}s}{1986}]{Bollobas86}
B\'{e}la Bollob\'{a}s.
\newblock {\em Combinatorics: Set Systems, Hypergraphs, Families of Vectors,
  and Combinatorial Probability}.
\newblock Cambridge University Press, 1986.

\bibitem[\protect\citeauthoryear{Calabro}{2004}]{Calabro04}
Chris Calabro.
\newblock Harper's theorem.
\newblock \url{http://cseweb.ucsd.edu/~ccalabro/essays/harper.pdf}, 2004.

\bibitem[\protect\citeauthoryear{Caragiannis \bgroup \em et al.\egroup
  }{2019}]{CaragiannisKuMo16}
Ioannis Caragiannis, David Kurokawa, Herv\'{e} Moulin, Ariel~D. Procaccia,
  Nisarg Shah, and Junxing Wang.
\newblock The unreasonable fairness of maximum {N}ash welfare.
\newblock {\em ACM Transactions on Economics and Computation},
  7(3):12:1--12:32, 2019.

\bibitem[\protect\citeauthoryear{Chakraborty \bgroup \em et al.\egroup
  }{2020}]{ChakrabortyIgSu20}
Mithun Chakraborty, Ayumi Igarashi, Warut Suksompong, and Yair Zick.
\newblock Weighted envy-freeness in indivisible item allocation.
\newblock In {\em Proceedings of the 19th International Conference on
  Autonomous Agents and Multi-Agent Systems (AAMAS)}, pages 231--239, 2020.

\bibitem[\protect\citeauthoryear{Chaudhury \bgroup \em et al.\egroup
  }{2020}]{ChaudhuryGaMe20}
Bhaskar~Ray Chaudhury, Jugal Garg, and Kurt Mehlhorn.
\newblock {EFX} exists for three agents.
\newblock {\em CoRR}, abs/2002.05119, 2020.

\bibitem[\protect\citeauthoryear{Kyropoulou \bgroup \em et al.\egroup
  }{2019}]{KyropoulouSuVo19}
Maria Kyropoulou, Warut Suksompong, and Alexandros~A. Voudouris.
\newblock Almost envy-freeness in group resource allocation.
\newblock In {\em Proceedings of the 28th International Joint Conference on
  Artificial Intelligence (IJCAI)}, pages 400--406, 2019.

\bibitem[\protect\citeauthoryear{Lipton \bgroup \em et al.\egroup
  }{2004}]{LiptonMaMo04}
Richard~J. Lipton, Evangelos Markakis, Elchanan Mossel, and Amin Saberi.
\newblock On approximately fair allocations of indivisible goods.
\newblock In {\em Proceedings of the 5th ACM Conference on Electronic Commerce
  (EC)}, pages 125--131, 2004.

\bibitem[\protect\citeauthoryear{Oh \bgroup \em et al.\egroup
  }{2019}]{OhPrSu19}
Hoon Oh, Ariel~D. Procaccia, and Warut Suksompong.
\newblock Fairly allocating many goods with few queries.
\newblock In {\em Proceedings of the 33rd AAAI Conference on Artificial
  Intelligence (AAAI)}, pages 2141--2148, 2019.

\bibitem[\protect\citeauthoryear{Plaut and Roughgarden}{2018}]{PlautRo18}
Benjamin Plaut and Tim Roughgarden.
\newblock Almost envy-freeness with general valuations.
\newblock In {\em Proceedings of the 29th Annual ACM-SIAM Symposium on Discrete
  Algorithms (SODA)}, pages 2584--2603, 2018.

\bibitem[\protect\citeauthoryear{Segal-Halevi and
  Suksompong}{2019}]{SegalhaleviSu19}
Erel Segal-Halevi and Warut Suksompong.
\newblock Democratic fair allocation of indivisible goods.
\newblock {\em Artificial Intelligence}, 277:103167, 2019.

\bibitem[\protect\citeauthoryear{Sperner}{1928}]{Sperner28}
Emanuel Sperner.
\newblock Ein {S}atz \"{u}ber {U}ntermengen einer endlichen {M}enge.
\newblock {\em Mathematische Zeitschrift}, 27(1):544--548, 1928.

\bibitem[\protect\citeauthoryear{Suksompong}{2019}]{Suksompong19}
Warut Suksompong.
\newblock Fairly allocating contiguous blocks of indivisible items.
\newblock {\em Discrete Applied Mathematics}, 260:227--236, 2019.

\end{thebibliography}

\end{document}